\documentclass[11pt]{article}

\usepackage[margin=1in]{geometry}
\usepackage{amsmath,amssymb,amsthm,mathtools}
\usepackage[authoryear]{natbib}
\usepackage[colorlinks=true,linkcolor=blue,citecolor=blue,urlcolor=blue]{hyperref}

\newtheorem{theorem}{Theorem}
\newtheorem{proposition}{Proposition}

\newtheorem{axiom}{Axiom}
\newtheorem{definition}{Definition}
\newtheorem{remark}{Remark}

\title{Single-Peakedness Does Not Prevent Leapfrogging under Abstention}

\author{%
Aman Ray\\
Madras School of Economics\\
\texttt{ba24aman@mse.ac.in}
\and
Srikanth B. Pai\\
Madras School of Economics\\
\texttt{srikanthbpai@mse.ac.in}\\
ORCID: 0000-0002-2283-8790
}

\date{\today}

\begin{document}

\maketitle

\begin{abstract}
Parties in spatial competition rarely choose platforms that reverse their ideological order. 
\textit{Mutual leapfrogging} is the strongest form of reversal: each party locates beyond the other party's ideal point. 
In voting models without abstention single-peakedness rules out such reversals.
We show that this conclusion does not survive endogenous abstention. 
There is a spatial voting model in which voter and party preferences are single-peaked, yet mutual leapfrogging occurs in pure-strategy equilibrium.
All voters participate at the equilibrium profile. 
The equilibrium survives because some deviations change which voters participate.
We prove that such equilibria are impossible under a sufficient ordinal condition: parties agree on how to rank leftward and rightward deviations from their ideal points. 
The condition is general enough to cover symmetric single-peaked utilities and common translated utility shapes.
\end{abstract}



\section{Introduction}
\label{sec:introduction}
Why do parties rarely reverse their ideological order and leapfrog one another?
\citet{adams2001} finds that parties rarely leapfrog one another and that their policy movements are constrained by ideology. Since \citet{Wittman1983}, models of ideological candidates have typically assume a concave utility function for ideological preferences. These assumptions imply that each party has an ideal point and ranks policies lower as they move away from it. This note asks whether this weaker ordinal property of single-peakedness is enough to rule out leapfrogging.

We first show that mutual leapfrogging is impossible without voter abstention. More generally, the same conclusion holds whenever the set of participating voters is independent of the parties' positions. A party that has crossed beyond its opponent can move back toward its own ideal point without losing electoral support, and the move is also ideologically better.

Abstention is therefore essential to the counterexample. We study alienation-style abstention: a voter participates if at least one party platform is acceptable to her. Once participation depends on the platforms, a party's deviation can change which voters participate. This is a standard reason for abstention in spatial voting models \citep{Anderson1992,Llavador2000}, and alienation is also an empirically relevant source of abstention \citep{Adams2006}. Once participation depends on the platforms, a party's deviation can change which voters participate. We first show that any equilibrium that reverses party order must exhibit \textit{mutual leapfrogging}: each party must have crossed the other party's ideal point.

We give an explicit counterexample showing that single-peaked party preferences alone do not rule out mutual leapfrogging under abstention. All voters participate at the equilibrium profile, but some deviations change which voters participate. 

The example identifies a criterion that enables leapfrogging. Single-peakedness ranks policies on each side of a party's ideal point, but not across sides. We then prove impossibility of mutual leapfrogging under a simple sufficient condition: parties agree on how to rank leftward and rightward deviations from their ideal points. This condition covers symmetric single-peaked utilities and common translated utility shapes, while remaining weaker than both.

Section \ref{sec:model} defines the model which is an ordinal version of alienation based abstention model of spatial voting. Section \ref{sec:counterexample} presents the example and shows why abstention matters. Section \ref{sec:crossside} presents the ordinal condition under which mutual leapfrogging is impossible.

\section{Model}
\label{sec:model}

Let
\[
X=\{x_j:j\in I\}, \qquad x_j<x_{j+1},
\]
where \(I\subseteq \mathbb Z\) is an interval of integers. Thus \(X\) may be finite or countably infinite. The integer index records position in the linear order: \(x_{j+k}\) is \(k\) steps to the right of \(x_j\), whenever \(j+k\in I\).

There are two parties, \(A\) and \(B\), with ideal points \(\tau_A=x_p\) and \(\tau_B=x_q\), where \(p<q\). Each party \(i\in\{A,B\}\) has a weak preference order \(\succeq_i\) over \(X\), with unique ideal point \(\tau_i\). We write \(x\succ_i y\) when \(x\succeq_i y\) holds while \(y\succeq_i x\) fails.

\begin{axiom}[Single-peakedness of party preferences]
\label{ax:ssm}
For each party \(i\in\{A,B\}\), the preference order \(\succeq_i\) is single-peaked with peak \(\tau_i\): if \(x_a<x_b\le \tau_i\), then \(x_b\succ_i x_a\); and if \(\tau_i\le x_b<x_a\), then \(x_b\succ_i x_a\).
\end{axiom}

Axiom \ref{ax:ssm} is the standard single-peakedness condition on a line. It ranks alternatives on the same side of the peak by proximity to the peak, but imposes no restriction on comparisons between alternatives lying on opposite sides of the peak.

\begin{remark}[Ordinal displacement notation]
\label{rem:ordinal_displacements}
For later use, we introduce signed ordinal-displacement notation. If \(\tau_i=x_c\), define \(R_i(k)=x_{c+k}\) and \(L_i(k)=x_{c-k}\), whenever these policies belong to \(X\). Thus \(R_i(k)\) and \(L_i(k)\) denote the policies \(k\) ordinal steps to the right and left of party \(i\)'s ideal point. By Axiom \ref{ax:ssm}, for each party \(i\),
\[
\tau_i\succ_i R_i(1)\succ_i R_i(2)\succ_i\cdots
\quad\text{and}\quad
\tau_i\succ_i L_i(1)\succ_i L_i(2)\succ_i\cdots,
\]
whenever the corresponding policies belong to \(X\). Axiom \ref{ax:ssm} imposes no restriction on comparisons of the form \(R_i(a)\) versus \(L_i(b)\).
\end{remark}

A party's strategy is a policy in \(X\). We refer to the policies chosen by the parties as their \emph{platforms}. If party \(A\) chooses \(s\in X\) and party \(B\) chooses \(t\in X\), we call \((s,t)\) a \emph{platform profile}.

\subsection*{Electoral outcomes}

We use an ordinal version of \citet{Llavador2000}'s alienation rule for a finite electorate.  Instead of requiring a party platform to lie within a fixed distance of the voter's ideal point, we require a party platform to lie in the voter's acceptable interval. 
These intervals may differ across voters.

The electorate is a finite set \(V\). Each voter \(v\in V\) has an ideal point \(\theta_v=x_{c_v}\in X\), strict single-peaked preferences over \(X\) with peak \(\theta_v\), and an attraction interval \(A_v\subseteq X\). We assume \(A_v\) is an interval in the policy order and contains \(\theta_v\).

Given platforms \((s,t)\), voter \(v\) is active if \(s\in A_v\) or \(t\in A_v\). If active, voter \(v\) votes for the platform she strictly prefers. If she is indifferent between the two platforms, she abstains. In particular, if \(s=t\), no active voter strictly prefers one party to the other.

% The electorate is a finite set \(V\). Each voter \(v\in V\) has an ideal point \(\theta_v=x_{c_v}\in X\), strict symmetric single-peaked preferences over \(X\) with peak \(\theta_v\), and an abstention radius \(k_v\in\mathbb Z_{\ge 0}\). Given platforms \((s,t)=(x_m,x_\ell)\), voter \(v\) is active if at least one platform lies within \(k_v\) ordinal steps of her ideal point:
% \[
% \min\{|m-c_v|,|\ell-c_v|\}\le k_v.
% \]
% If active, voter \(v\) votes for the platform she strictly prefers. If she is indifferent between the two platforms, she abstains. In particular, if \(s=t\), no active voter strictly prefers one party to the other.

Let \(N_A(s,t)\) be the number of active voters who strictly prefer \(s\) to \(t\), and let \(N_B(s,t)\) be the number of active voters who strictly prefer \(t\) to \(s\). Define \(g:X\times X\to\{A,B,T\}\) by
\[
g(s,t)=
\begin{cases}
A, & N_A(s,t)>N_B(s,t),\\
B, & N_B(s,t)>N_A(s,t),\\
T, & N_A(s,t)=N_B(s,t).
\end{cases}
\]
Here \(T\) denotes a tied electoral outcome. Since \(g(s,s)=T\) for every \(s\in X\), copying the opponent's platform always produces a tie. This is the only property of the voting rule used in Proposition \ref{prop:order_reversal}.

\subsection*{Party objectives}

Parties first rank electoral outcomes:
\[
\mathrm{Win}\succ \mathrm{Tie}\succ \mathrm{Lose}.
\]
For party \(A\), outcome \(A\) is a win and outcome \(B\) is a loss; for party \(B\), outcome \(B\) is a win and outcome \(A\) is a loss.

Conditional on the same electoral outcome, each party ranks profiles by its own announced platform. Thus, if \(g(s,t)=g(s',t')\), then party \(A\) weakly prefers \((s,t)\) to \((s',t')\) if and only if \(s\succeq_A s'\), and party \(B\) weakly prefers \((s,t)\) to \((s',t')\) if and only if \(t\succeq_B t'\).

This specification avoids any cardinal tradeoff between electoral success and ideology. Ideological preferences are used only to compare profiles with the same electoral outcome. The specification is nevertheless office-first: a win is always preferred to a tie, and a tie is always preferred to a loss.

A pure-strategy Nash equilibrium is a platform profile \((s,t)\in X\times X\) such that neither party has a strictly profitable unilateral deviation under this lexicographic ordering.

\begin{definition}
A platform profile \((s,t)\) exhibits \emph{reversed order} if \(t<s\). It exhibits \emph{mutual leapfrogging} if each party crosses the other party's ideal point, that is, if
\[
t<\tau_A<\tau_B<s.
\]
A Nash equilibrium exhibits reversed order, respectively mutual leapfrogging, if its equilibrium profile does.
\end{definition}

Since \(\tau_A<\tau_B\), reversed order means that the left party chooses the right platform and the right party chooses the left platform. The next proposition shows that any reversed-order equilibrium must be mutually leapfrogged.

\begin{proposition}[Order reversal implies mutual leapfrogging]
\label{prop:order_reversal}
Suppose Axiom \ref{ax:ssm} holds. Let \((s,t)\) be a pure-strategy Nash equilibrium with \(t<s\). Then the election is tied at \((s,t)\), and
\[
t<\tau_A<\tau_B<s.
\]
\end{proposition}

\begin{proof}
The election must be tied. If one party loses, it can copy the opponent's platform and obtain a tie, which is strictly better than losing. Since party \(A\)'s deviation to \(t\) gives a tie, equilibrium requires \(s\succeq_A t\). If \(\tau_A\le t<s\), then \(t\) is closer than \(s\) to \(\tau_A\) on the same side, so single-peakedness gives \(t\succ_A s\), a contradiction. Hence \(t<\tau_A\).
Similarly, party \(B\)'s deviation to \(s\) gives a tie, so equilibrium requires \(t\succeq_B s\). If \(t<s\le \tau_B\), then \(s\) is closer than \(t\) to \(\tau_B\) on the same side, so single-peakedness gives \(s\succ_B t\), a contradiction. Hence \(\tau_B<s\).
Since \(\tau_A<\tau_B\), the result follows.
\end{proof}

Proposition \ref{prop:order_reversal} uses only the fact that copying the opponent's platform produces a tie. The existence result below uses abstention more substantively: the set of active voters can change after deviations.

\section{Single-peakedness alone is not sufficient with abstention}
\label{sec:counterexample}

We first show that mutual leapfrogging is impossible when the set of participating voters does not change with the pair of platforms. We then give a simple example in which mutual leapfrogging occurs with four voters and seven policies: all voters participate at the equilibrium profile, but some deviations change who participates.

\begin{proposition}
\label{prop:same_active_voters_no_leapfrogging}
Suppose the same voters are active at every pair of distinct platforms. Then no pure-strategy Nash equilibrium exhibits mutual leapfrogging.
\end{proposition}

\begin{proof}
Suppose, toward a contradiction, that \((s,t)\) is a pure-strategy Nash equilibrium exhibiting mutual leapfrogging. Then \(t<\tau_A<\tau_B<s\).

The election at \((s,t)\) must be tied. If party \(A\) loses, it can copy party \(B\)'s platform and obtain a tie, which is strictly better than losing. The same argument applies to party \(B\). Hence the election is tied.

We claim that party \(A\)'s deviation from \(s\) to \(\tau_A\) cannot reduce its number of votes. Let \(v\) be any voter who votes for \(s\) against \(t\). Since \(t<\tau_A<s\), single-peakedness of voter preferences implies that \(v\) also prefers \(\tau_A\) to \(t\). Indeed, if \(v\)'s peak lies weakly to the left of \(t\), then \(t\succ_v s\), a contradiction. If the peak lies between \(t\) and \(\tau_A\), then \(\tau_A\succ_v s\), hence \(\tau_A\succ_v t\). If the peak lies weakly to the right of \(\tau_A\), then \(\tau_A\succ_v t\).

By assumption, the same voters are active at \((\tau_A,t)\) as at \((s,t)\). Therefore every voter counted for party \(A\) at \((s,t)\) is still counted for party \(A\) at \((\tau_A,t)\). Since \((s,t)\) is tied, party \(A\) obtains either a tie or a win after deviating to \(\tau_A\).

If it wins, the deviation is profitable. If it ties, the electoral outcome is unchanged and \(\tau_A\succ_A s\), since \(\tau_A\) is party \(A\)'s unique ideal point. In either case party \(A\) has a profitable deviation, contradicting equilibrium.
\end{proof}

The example below shows what abstention changes. At the equilibrium profile all voters are active, but deviations can change which voters are active. The example uses seven policies and four voters.

Let \(X=\{x_1,\ldots,x_7\}\), with \(x_1<\cdots<x_7\). The parties' ideal points are \(\tau_A=x_3\) and \(\tau_B=x_5\). Party preferences are strict and given by
\[
\begin{array}{ll}
A: & x_3\succ_A x_4\succ_A x_5\succ_A x_6\succ_A x_2\succ_A x_1\succ_A x_7,\\[1mm]
B: & x_5\succ_B x_4\succ_B x_3\succ_B x_2\succ_B x_6\succ_B x_1\succ_B x_7.
\end{array}
\]
Both preferences are single-peaked. Party \(A\)'s right side is ranked \(x_3\succ_A x_4\succ_A x_5\succ_A x_6\succ_A x_7\), and its left side is ranked \(x_3\succ_A x_2\succ_A x_1\). Party \(B\)'s left side is ranked \(x_5\succ_B x_4\succ_B x_3\succ_B x_2\succ_B x_1\), and its right side is ranked \(x_5\succ_B x_6\succ_B x_7\).

There are four voters, two on each side. Their ideal points are \(\theta_{v_1}=x_1\), \(\theta_{v_2}=x_2\), \(\theta_{v_6}=x_6\), and \(\theta_{v_7}=x_7\). Their attraction intervals are
\[
A_{v_1}=\{x_1,x_2\},\qquad
A_{v_2}=\{x_1,x_2,x_3\},
\]
and
\[
A_{v_6}=\{x_5,x_6,x_7\},\qquad
A_{v_7}=\{x_6,x_7\}.
\]
Their preferences are strict and single-peaked:
\[
\begin{array}{ll}
v_1: & x_1\succ x_2\succ x_3\succ x_4\succ x_5\succ x_6\succ x_7,\\
v_2: & x_2\succ x_1\succ x_3\succ x_4\succ x_5\succ x_6\succ x_7,\\
v_6: & x_6\succ x_5\succ x_7\succ x_4\succ x_3\succ x_2\succ x_1,\\
v_7: & x_7\succ x_6\succ x_5\succ x_4\succ x_3\succ x_2\succ x_1.
\end{array}
\]

At the profile \((s,t)=(x_6,x_2)\), all four voters are active. Voters \(v_1\) and \(v_2\) vote for party \(B\), while voters \(v_6\) and \(v_7\) vote for party \(A\). Hence \(g(x_6,x_2)=T\). The unilateral deviations are summarized in the following table:
\[
\begin{array}{c|ccccccc}
 & x_1 & x_2 & x_3 & x_4 & x_5 & x_6 & x_7\\
\hline
g(\cdot,x_2) & T & T & B & B & B & T & T\\
g(x_6,\cdot) & T & T & A & A & A & T & T
\end{array}
\]

Against \(x_2\), party \(A\)'s non-losing platforms are \(x_1,x_2,x_6,x_7\), and \(A\) strictly prefers \(x_6\) to all three alternatives. Against \(x_6\), party \(B\)'s non-losing platforms are \(x_1,x_2,x_6,x_7\), and \(B\) strictly prefers \(x_2\) to all three alternatives. Thus \((x_6,x_2)\) is a pure-strategy Nash equilibrium. Since
\[
x_2<\tau_A=x_3<\tau_B=x_5<x_6,
\]
the equilibrium exhibits mutual leapfrogging.

\section{Agreement on cross-side comparisons}
\label{sec:crossside}

Single-peakedness compares policies lying on the same side of a party's ideal point. It does not compare a rightward deviation with a leftward deviation. The counterexample exploits precisely this freedom. We now impose an ordinal restriction on such cross-side comparisons. We use the notation \(R_i(k)\) and \(L_i(k)\) from Remark \ref{rem:ordinal_displacements}.

\begin{axiom}[Agreement on cross-side comparisons]
\label{ax:cpc}
For all \(a,b\in\mathbb Z_{>0}\) such that the relevant policies belong to \(X\),
\[
R_A(a)\succeq_A L_A(b)
\quad\Longleftrightarrow\quad
R_B(a)\succeq_B L_B(b),
\]
and
\[
L_A(b)\succeq_A R_A(a)
\quad\Longleftrightarrow\quad
L_B(b)\succeq_B R_B(a).
\]
\end{axiom}

Axiom \ref{ax:cpc} says that the parties agree on the ranking of each cross-side pair: a rightward deviation of size \(a\) and a leftward deviation of size \(b\). It does not require left-right symmetry. It allows both parties to prefer, say, a rightward deviation of size \(a\) to a leftward deviation of size \(b\). It rules out only disagreement between the parties over the same cross-side pair.

The axiom is ordinal, but it is implied by familiar cardinal restrictions.

\begin{proposition}
\label{prop:axiom2}
Axiom \ref{ax:cpc} is implied by each of the following assumptions.

\begin{enumerate}
\item[(i)] \textbf{Symmetric single-peaked utilities.} For each party \(i\), there is a strictly decreasing function \(\phi_i:\mathbb Z_{\ge 0}\to\mathbb R\) such that, if \(\tau_i=x_c\), then \(v_i(x_m)=\phi_i(|m-c|)\).

\item[(ii)] \textbf{Common utility shape.} There is a function \(u:\mathbb Z\to\mathbb R\) such that, if \(\tau_i=x_c\), then \(v_i(x_m)=u(m-c)\).
\end{enumerate}
\end{proposition}

\begin{proof}
Under (i), for each party \(i\),
\[
R_i(a)\succeq_i L_i(b)
\quad\Longleftrightarrow\quad
\phi_i(a)\ge \phi_i(b)
\quad\Longleftrightarrow\quad
a\le b,
\]
and
\[
L_i(b)\succeq_i R_i(a)
\quad\Longleftrightarrow\quad
\phi_i(b)\ge \phi_i(a)
\quad\Longleftrightarrow\quad
b\le a.
\]
Both conditions are independent of \(i\), so the parties agree on every cross-side pair.

Under (ii), for each party \(i\),
\[
R_i(a)\succeq_i L_i(b)
\quad\Longleftrightarrow\quad
u(a)\ge u(-b),
\]
and
\[
L_i(b)\succeq_i R_i(a)
\quad\Longleftrightarrow\quad
u(-b)\ge u(a).
\]
Again, the right-hand sides are independent of \(i\). Hence Axiom \ref{ax:cpc} holds.
\end{proof}

Axiom \ref{ax:cpc} is weaker than either sufficient condition above. Symmetric distance utilities imply \(R_i(k)\sim_i L_i(k)\) for every party \(i\) and every \(k\). Axiom \ref{ax:cpc} does not require this symmetry. It allows both parties to prefer, for example, a rightward deviation of two steps to a leftward deviation of one step, provided they agree on that comparison. A common translated utility shape is also stronger: it imposes a common numerical representation over all signed deviations. Axiom \ref{ax:cpc} keeps only agreement on the ordinal ranking of each cross-side pair.

We now show that this ordinal agreement rules out mutual leapfrogging.

\begin{theorem}
\label{thm:no_leapfrogging}
Suppose Axioms \ref{ax:ssm} and \ref{ax:cpc} hold. Then no pure-strategy Nash equilibrium exhibits mutual leapfrogging.
\end{theorem}

\begin{proof}
Suppose, toward a contradiction, that \((s,t)\) is a pure-strategy Nash equilibrium exhibiting mutual leapfrogging. Then \(t<\tau_A<\tau_B<s\). By Proposition \ref{prop:order_reversal}, the election is tied at \((s,t)\). Since party \(A\)'s deviation to \(t\) also gives a tie, equilibrium requires \(s\succeq_A t\). Since party \(B\)'s deviation to \(s\) also gives a tie, equilibrium requires \(t\succeq_B s\).

Write \(\tau_A=x_p\), \(\tau_B=x_q\), \(t=x_r\), and \(s=x_u\), with \(r<p<q<u\). Let \(\delta=q-p>0\), \(a=u-p\), and \(b=p-r\). Relative to party \(A\)'s ideal point, \(s=R_A(a)\) and \(t=L_A(b)\). Thus \(s\succeq_A t\) gives \(R_A(a)\succeq_A L_A(b)\). By Axiom \ref{ax:cpc}, \(R_B(a)\succeq_B L_B(b)\).

Relative to party \(B\)'s ideal point,
\[
s=R_B(a-\delta),
\qquad
t=L_B(b+\delta).
\]
By Axiom \ref{ax:ssm}, \(R_B(a-\delta)\succ_B R_B(a)\) and \(L_B(b)\succ_B L_B(b+\delta)\). Therefore
\[
s=R_B(a-\delta)\succ_B R_B(a)\succeq_B L_B(b)\succ_B L_B(b+\delta)=t.
\]
Hence \(s\succ_B t\), contradicting \(t\succeq_B s\). Therefore no pure-strategy Nash equilibrium exhibits mutual leapfrogging.
\end{proof}
\begin{remark}
Axiom \ref{ax:cpc} is sufficient, not necessary, for ruling out mutual leapfrogging. In the example of Section \ref{sec:counterexample}, if every voter regards every policy as acceptable, then the same voters are active at every pair of distinct platforms. Proposition \ref{prop:same_active_voters_no_leapfrogging} then rules out mutual leapfrogging. The party preferences are unchanged, however, and Axiom \ref{ax:cpc} still fails.
\end{remark}
\section{Conclusion}
\label{sec:conclusion}

Single-peakedness ranks policies on each side of a party's ideal point, but not across sides. With abstention, this is enough to permit mutual leapfrogging in pure-strategy equilibrium. Agreement on cross-side comparisons rules it out. Single-peakedness alone is therefore not enough to prevent party-order reversal.

\bibliographystyle{plainnat}
\bibliography{cas-refs}

\end{document}